\documentclass[10pt, twocolumn]{IEEEtran}

\usepackage{epsfig,latexsym}
\usepackage{amsmath}
\usepackage{amssymb}
\usepackage{subfigure}
\usepackage[colorlinks]{hyperref}
\usepackage{indentfirst}
\usepackage{times}
\usepackage{fancyhdr}
\usepackage{lastpage}
\usepackage{bigints}
\usepackage{subfigure}
\usepackage{graphicx}
\usepackage{amsthm}

\newtheorem{corollary}{Corollary}
\newtheorem{lemma}{Lemma}

\usepackage[noend]{algpseudocode}
\usepackage{algorithmicx,algorithm}
\usepackage{color}
\usepackage{ulem}

\begin{document}
\title{Stochastic Geometry Analysis of Uplink CUMA-Enabled Cellular Networks}

\author{Chenguang Rao, 
Kai-Kit Wong~\IEEEmembership{Fellow,~IEEE},
Xusheng Zhu,
Hanjiang Hong,~\IEEEmembership{Member,~IEEE},\\ 
Chan-Byoung Chae, \emph{Fellow, IEEE}, and 
Ross Murch, \emph{Fellow, IEEE}
\vspace{-8mm}

\thanks{The work of K. K. Wong is supported by the Engineering and Physical Sciences Research Council (EPSRC) under grant EP/W026813/1.}
\thanks{The work of C.-B. Chae was supported by the Korean Government under Grant IITP-2025-RS-2024-00428780 and Grant IITP-RS-2026-25489110.}
\thanks{The work of R. Murch was supported by the Hong Kong Research Grants Council Area of Excellence Grant AoE/E-601/22-R.}

\thanks{C. Rao, K. K. Wong, X. Zhu, and H. Hong are with the Department of Electronic and Electrical Engineering, University College London, London WC1E 7JE, United Kingdom. K. K. Wong is also affiliated with Yonsei Frontier Laboratory, Yonsei University, Seoul, 03722, Republic of Korea (e-mail:  \{chenguang.rao,kai\text{-}kit.wong,xusheng.zhu,hanjiang.hong\}@ucl.ac.uk).}
\thanks{C. B. Chae is with School of Integrated Technology, Yonsei University, Seoul, 03722, Republic of Korea (e-mail: cbchae@yonsei.ac.kr).}
\thanks{R. Murch is with the Department of Electronic and Computer Engineering, Hong Kong University of Science and Technology, Clear Water Bay, Hong Kong SAR, China (e-mail: eermurch@ust.hk).}

\thanks{Corresponding author: Kai-Kit Wong.}
}

\maketitle
\begin{abstract}
Uplink cellular networks are interference-dominated but interference channel state information (CSI) is rarely available at scale. The emerging fluid antenna system (FAS) concept, which provides additional spatial degrees of freedom through multi-port reconfiguration, offers a promising alternative to CSI-intensive multi-antenna processing. Building on this concept, compact ultra-massive arrays (CUMA) exploit large-scale port selection with low implementation complexity. In each uplink transmission, CUMA activates a subset of ports based on only the desired-link CSI and combines the selected ports via simple superposition, yielding coherent enhancement of the desired user signal, while inter-cell interference aggregates largely non-coherently due to the random superposition effect. Consequently, CUMA is well suited to multi-cell uplink scenarios where CSI is limited. In this paper, we analyze uplink CUMA in multi-cell cellular networks using a stochastic geometry framework. We derive a tight approximate expression for the signal-to-interference ratio (SIR) coverage probability, and further characterize the average user rate and cell sum-rate. The analysis quantifies how key design parameters impact performance and reveals the scaling behavior with network densification. Simulation results validate the accuracy of the derived expressions and show that uplink CUMA achieves competitive, and often superior, performance relative to conventional schemes under practical CSI constraints, highlighting its potential as a low-complexity, hardware-efficient uplink solution for future large-scale cellular networks.
\end{abstract}

\begin{IEEEkeywords}
Compact ultra-massive array (CUMA), fluid antenna system (FAS), fluid antenna multiple access (FAMA), stochastic geometry, cellular network, uplink. 
\end{IEEEkeywords}

\vspace{-2mm}
\section{Introduction}\label{sec:intro}
\IEEEPARstart{C}{ellular} networks have consistently been a key technology in wireless communication research due to their ability to provide wide-area coverage, efficient spectrum reuse, and flexible multiple access \cite{CN1,CN2,CN3}. As wireless systems continue to evolve from fifth generation (5G) toward sixth generation (6G) and beyond, cellular networks are expected to support extreme massive connectivity, hyper-reliable low-latency transmission, immersive services, and a broad range of intelligent applications \cite{6G1,6G2,Tariq-2020,6G4}. In such a vision, the spatial density of both base stations (BSs) and user equipments (UEs) will continue to increase. These trends make cellular networks more powerful, but also intensify one of their most challenges: significant and non-negligible interference. In cellular communications, multiple UEs are likely served with one BS over limited spectral resources. Since aggressive frequency reuse is essential for achieving high spectral efficiency, the signal at each transmission point is inevitably corrupted not only by intra-cell interference from UEs in the same cell, but also by inter-cell interference (ICI) generated by active UEs in neighboring cells. In dense deployments, such interference may become the dominant performance bottleneck. This problem is particularly pronounced in multi-cell uplink systems, where the network must simultaneously accommodate many distributed UEs with heterogeneous channel conditions \cite{CNI}. Thus, an important question in uplink cellular design is how to achieve reliable and effective multiuser reception under strong interference, while keeping the signal processing complexity and hardware burden within practical limits.

To deal with this issue, a wide variety of receiver architectures and network paradigms have been investigated. In conventional multiuser multiple-input multiple-output (MIMO) systems, linear receiver methods such as maximal-ratio combining (MRC), zero-forcing (ZF), regularized zero-forcing (RZF), and minimum mean-square error (MMSE) receivers are widely adopted because of their trade-offs between complexity and interference suppression capability \cite{CV1,CV2}. In addition to linear processing techniques, other approaches have also been explored. For instance, non-orthogonal multiple access (NOMA) and rate-splitting multiple access (RSMA) schemes employ successive interference cancellation (SIC) at the receiver to improve spectral efficiency by decoding users in a specific order \cite{NOMA,RSMA}. These approaches aim to exploit additional spatial or decoding degrees of freedom to further enhance interference management performance. 

In massive MIMO systems, more advanced variants such as multi-cell MMSE have also been studied to exploit broader channel situation information (CSI) for enhanced interference mitigation \cite{CV3}. Beyond the conventional approaches, cell-free massive MIMO has emerged as another potential solution, in which many distributed access points cooperatively serve the users simultaneously thereby improving spatial diversity and reducing cell-edge performance degradation. Despite their significant advantages, the above solutions also face limitations when the network scale becomes large and the interference environment is highly complicated. Specifically, most of these schemes rely on sufficiently accurate instantaneous CSI, and their performance can be sensitive to pilot overhead, pilot contamination, channel estimation errors, and CSI exchange across processing units \cite{ICI1,ICI2}. Furthermore, these schemes usually require high computational complexity and hardware burden. For example, MMSE-based receivers typically rely on a large number of radio-frequency (RF) chains, while cell-free massive MIMO requires dense deployment of distributed access points and extensive coordination among them. As a result, while these approaches are powerful, there is strong motivation to seek an alternative uplink reception strategy that retains robustness against interference but rely on much less CSI, simpler hardware, and lower processing complexity.

Recently, fluid antenna system (FAS) has attracted increasing research interest as an emerging technology that introduces new spatial freedom into wireless communications \cite{wong2020_fas_limits,wong2021_fas_twc}. Different from conventional fixed-position antennas, an FAS provides multiple candidate ports within a compact physical aperture and allows the communication devices to dynamically activate favorable ports according to the instantaneous channel condition \cite{FAS1,FAS2,FAS3,FAS4,hong2026_fas_survey,FAS_wu_tuo1,FAS5}. The FAS concept is hardware-agnostic but is greatly motivated by recent advances in mechanically movable elements \cite{zhu2024_fas_history}, liquid-based antennas \cite{shen2024_surfacewave_fas,wang2026_em_reconfig_fas}, metamaterial-based apertures \cite{Zhang-jsac2026,liu2025_meta_fluid_optics}, and pixel-based structures \cite{zhang2025_pixel_reconfig,liu2025_wideband_pixel_fas,wong2026_pixel_meet_fas}. With reconfigurability in antenna position, fluid antenna multiple access (FAMA) was developed to enable interference mitigation \cite{FAMA3}. 

The basic idea of FAMA is to exploit the spatial fluctuation of the received signals across different ports, such that the receiver can select appropriate ports where the desired signal is enhanced while the interference is weakened. In its early form, often referred to as fast FAMA, the receiver performs port selection on a per-symbol basis to maximize the instantaneous signal-to-interference-plus-noise ratio (SINR), which achieves significant performance gains but requires frequent switching and high complexity \cite{wong2022_fama,FAMA1}. To improve practicality, slow FAMA was later proposed, where port selection is performed over a much longer time scale, such as per channel coherence block, thereby reducing hardware and signaling burden while retaining part of the performance benefits \cite{FAMA2,FAMA6}. Building upon these foundations, various enhanced FAMA techniques have been developed via channel coding \cite{hong2025coded,hong20255gcoded,FAMA4} and learning-based data processing \cite{waqar2026_turbocharging_fama,FAMA7} to further improve performance and adaptability. A recent work in \cite{FAMA8} further demonstrates that the BS can operate without CSI for multiple access in the uplink if UEs are each equipped with FAS.

Comparing the various FAMA schemes, it is always about balancing between multiplexing capability and complexity. In this regard, the compact ultra-massive array (CUMA) architecture is often considered the best approach \cite{CUMA1}, as it keeps the complexity and CSI requirement low, similar to slow FAMA, and achieves much higher multiplexing gain than slow FAMA. The key principle of CUMA is to focus on maximizing the desired signal power without explicitly considering interference. Based on this idea, CUMA activates a subset of ports and combines them through a simple superposition rule, thereby achieving constructive signal enhancement with low hardware complexity. Meanwhile, since the port selection depends only on the desired channel, the interference signals experience essentially random superposition across the selected ports and as a result, do not benefit from the signal combining in the analog domain. This design avoids the requirement for interference awareness or processing, such as instantaneous CSI, multiuser decoding, or complex beamforming optimization, and thus significantly reduces both computational complexity and hardware requirements. Recent studies have considered enhancements for CUMA \cite{CUMA2,CUMA3,CUMA4}, illustrating strong potential in supporting massive access networks. 

A recently work in \cite{CUMA_SG} has already introduced the concept of uplink CUMA in a multi-cell environment and showed by simulation results that CUMA can provide strong interference resilience while requiring only low-complexity receiver operations and limited CSI. A particularly attractive feature of uplink CUMA is that its port-selection mechanism depends only on the desired user's channel observed at the serving BS, rather than on the instantaneous CSI of out-of-cell interferers. Nevertheless, although the mentioned work has revealed the practical potential of uplink CUMA, it still leaves an important issue: the lack of analytical performance characterization. At present, the benefits of uplink CUMA in multi-cell networks are mainly supported by simulation evidence, while the statistical behaviors of key metrics such as the desired signal power, the SINR, and the achievable rate remain insufficiently studied. Without such analytical results, it is difficult to understand how the performance scales with device or environment parameters, such as the number of antenna ports, the physical aperture size, the BS and UE densities, or the path loss exponents. It is also difficult to compare the gains of uplink CUMA with other network designs from a more general and systematic perspective. Therefore, developing a tractable analytical framework is not only a theoretical exercise, but an important step toward establishing uplink CUMA as a practical design.

Among the available analytical tools, stochastic geometry provides a particularly suitable framework for the problem considered in this paper. By modeling BS and UE locations as spatial point processes, it facilitates the analysis of wireless networks with random topology and captures the effect of network-wide interference \cite{SG1}. Since its tractability and effectiveness in characterizing coverage and rate in cellular systems have been well recognized, stochastic geometry has become a key tool for the study of large-scale wireless networks \cite{SG2,SG3}. For the present problem in cellular networks, stochastic geometry is attractive for two reasons. First, the key difficulty in uplink CUMA lies not only in the small-scale fading and port correlation, but also in the random spatial distribution of interferers across the entire network. Secondly, stochastic geometry has been widely employed to analyze uplink cellular networks with spatially distributed users, where it provides a tractable framework to characterize aggregate interference and evaluate key performance metrics such as coverage probability and average rate \cite{SG4,SG5}. Therefore, it serves as a natural and powerful tool for studying uplink CUMA in large-scale networks. Accordingly, this work adopts a stochastic geometry model to analyze uplink CUMA in multi-cell cellular networks and to derive tractable expressions that can reveal the impact of key system parameters. Specifically, we consider an uplink cellular network where BSs and UEs are distributed according to independent homogeneous Poisson point processes (PPPs), and each UE is associated with its nearest BS. Moreover, each BS is equipped with a two-dimensional (2D) FAS consisting of multiple ports over a compact aperture, while each UE employs a single fixed-position antenna. Assuming Rayleigh fading due to rich scattering, we investigate the uplink CUMA system where each BS selects and superposes ports according to the desired user's channel. The main analytical challenge is that the performance depends jointly on the random service distance, the spatially correlated port-domain fading, the randomness of the CUMA port-selection rule, and the aggregate interference from all other active UEs in the network. To overcome this, we separately characterize the desired signal and the interference statistics, and then combine them into a tractable approximation of the signal-to-interference (SIR) distribution. Our main contributions are summarized as follows:
\begin{itemize}
\item First and foremost, we establish an analytical stochastic geometry-based framework for uplink CUMA in multi-cell networks. Under the nearest-BS association rule and correlated Rayleigh fading across fluid antenna ports, we characterize the desired-signal power at the CUMA output and approximate its distribution in a tractable form. In parallel, we derive the statistical characterization of the aggregate interference by exploiting the random structure of the CUMA port-selection rule together with the probability generating functional of PPPs.
\item Based on the obtained distributions of the desired signal and interference, we derive an approximate expression for the SIR coverage probability of uplink CUMA. Building upon this result, we further analyze the average rate, and obtain a useful asymptotic characterization in the high-load regime. These analytical results reveal how the performance depends on the system parameters.
\item Furthermore, we provide simulation results to validate the developed analytical expressions and to demonstrate the effectiveness of uplink CUMA in interference-limited multi-cell scenarios. The results confirm that the proposed analysis is accurate and can capture the main performance trends with respect to key system and propagation parameters. In addition, by comparing with conventional reception schemes, the results highlight the performance advantage of uplink CUMA under practical CSI constraints, and illustrate the potential of uplink CUMA as a low-complexity and low-CSI uplink reception solution for future large-scale cellular networks.
\end{itemize}

The rest of this paper is organized as follows. Section \ref{sec:model} introduces the stochastic geometry-based system model and the uplink CUMA scheme. Section \ref{sec:analysis} presents the performance analysis. Section \ref{sec:result} then provides simulation results to verify the analysis and to illustrate the impact of key parameters. Finally, Section \ref{sec:conclude} concludes the paper. The main symbols used throughout this paper are summarized in Table~\ref{tab:notation}.

\begin{table}[!t]
\renewcommand{\arraystretch}{1.3}
\caption{Notations}\label{tab:notation}
\centering
\begin{tabular}{ll}
		\hline
		\textbf{Notation} & \textbf{Description} \\
		\hline
		$x$, $\mathbf{x}$, $\mathbf{X}$ & Scalar, vector, and matrix \\
		$X_{m,n}$ or $[{\bf X}]_{m,n}$ & The $(m,n)$-th entry of matrix $\mathbf{X}$ \\
		$(\cdot)^{T}$, $(\cdot)^{H}$ & Transpose, Hermitian transpose \\
		$\Phi_{i}$ &  PPP of i\\
		$\mathbf{e}_n$ &  The standard basis vector\\
		$\Re\{\cdot\}$ & Real part \\
		$\mathbb{E}\{\cdot (\mid \cdot)\}$ & (Conditioned) expectation \\
		$\mathrm{var}(\cdot)$, $\mathrm{cov}(\cdot,\cdot)$ & Variance, covariance \\
		$j_{0}(\cdot)$ & Zeroth-order Bessel function of the first kind \\
		$\bigcup_{\cdot}, \cdot\setminus\cdot $ & Union and set difference operator \\
		${}_{2}F_{1}(\cdot)$ & Gaussian hypergeometric function \\
		$\Gamma(\cdot)$ & Gamma function \\
		$\mathcal{L}_{\cdot (\mid \cdot)}(\cdot (\mid \cdot))$ &  (Conditioned) Laplace transform \\
		$\binom{\cdot}{\cdot}$ & Binomial coefficient \\
		$\lceil \cdot \rceil$ & Ceiling operator \\
		$\mathcal{O}(\cdot)$ & The asymptotic order of a function \\
		$\mathcal{(C)N}(\cdot,\cdot)$ & (Complex) Gaussian distribution \\
		$\mathrm{Gamma}\left(\cdot,\cdot\right)$ & Gamma distribution \\
		$\mathcal{X}_1$ & Standard Chi-squared distribution\\
		\hline
\end{tabular}
\end{table}

\vspace{-2mm}
\section{System Model and Uplink CUMA}\label{sec:model}
We first build a system model based on stochastic geometry, and then introduce the uplink CUMA briefly.

\vspace{-2mm}
\subsection{System Model}
Consider an uplink cellular network consisting of multiple BSs and UEs, which is modeled using stochastic geometry. To model the locations of all BSs, we use a PPP \(\Phi_b \subset \mathbb{R}^2\) of density \(\lambda_b\) in \(\mathbb{R}^2\). Specifically, for any bounded region \(\mathcal{A} \subset \mathbb{R}^2\), the number of BSs located in \(\mathcal{A}\), denoted by \(N_b(\mathcal{A})\), follows a Poisson distribution, given by
\begin{equation}
\mathbb{P}\left(N_b(\mathcal{A}) = k\right)= \frac{(\lambda_b |\mathcal{A}|)^k}{k!}\exp\left(-\lambda_b |\mathcal{A}|\right),
\end{equation}
where \(|\mathcal{A}|\) represents the area of \(\mathcal{A}\). Moreover, the numbers of BSs located in disjoint regions are independent. Similarly, the locations of all UEs are modeled as a PPP \(\Phi_u \subset \mathbb{R}^2\) of density \(\lambda_u\). Each UE has a single fixed antenna, while each BS is equipped with a 2D FAS of size \(W_1 \lambda \times W_2 \lambda\), which contains \(N = N_1 \times N_2\) ports, with \(\lambda\) being the carrier wavelength. 

Before communication, each UE associates with its nearest BS, which partitions the plane into Voronoi cells. Without loss of generality, by invoking Slivnyak's theorem \cite{SG1}, we focus on a typical BS located at the origin, denoted by \(\mathrm{BS}_0\). Let \(\mathcal{U}\) and \(\mathcal{U}_0\) denote the sets of all UEs and those associated with \(\mathrm{BS}_0\), respectively. In this paper, we adopt a Rayleigh fading model without line-of-sight (LoS) propagation. Then we define the channel between \(\mathrm{BS}_0\) and \(\mathrm{U}_u \in \mathcal{U}\) as
\begin{equation}
\mathbf{h}_u= \sqrt{\beta_u}\, \widetilde{\mathbf{h}}_u,
\end{equation}
where \(\beta_u = r_u^{-\eta}\) represents the large-scale path loss with scale \(\eta > 2\), \(r_u\) denotes the propagation distance between \(\mathrm{BS}_0\) and \(\mathrm{U}_u\), and \(\widetilde{\mathbf{h}}_u\) represents the small-scale fading, which follows the correlated complex Gaussian distribution:
\begin{equation}
\widetilde{\mathbf{h}}_u = \{\widetilde{h}_u^{(1)},\widetilde{h}_u^{(2)},\dots,\widetilde{h}_u^{(N)}\} \sim \mathcal{CN}\left(\mathbf{0}, \mathbf{J}\right),
\end{equation}
where \(\mathbf{J} \in \mathbb{R}^{N \times N}\) denotes the spatial correlation matrix across all ports. Define an index mapping \((n_1,n_2) \rightarrow k_{n_1,n_2}\):
\begin{equation}
k_{n_1,n_2} = n_2 + (n_1 - 1) N_2,~1\leq n_1\leq N_1, n_2\leq N_2.
\end{equation}
Following \cite{FAS3}, the correlation coefficient between \(h_u^{(k_{n_1,n_2})}\) and \(h_u^{(k_{\widetilde{n}_1,\widetilde{n}_2})}\) is given by
\begin{multline}\label{eq:Jk}
J_{k_{n_1,n_2},k_{\widetilde{n}_1,\widetilde{n}_2}}= \\ 
j_0\left(2\pi\sqrt{\left(\frac{n_1 - \widetilde{n}_1}{N_1-1}W_1\right)^2+\left(\frac{n_2 - \widetilde{n}_2}{N_2-1}W_2\right)^2}\right),
\end{multline}
where \(j_0(\cdot)\) is the zeroth-order Bessel function of the first kind.

During the transmission period, all UEs transmit their messages simultaneously over the same time-frequency resource block in the uplink. In this paper, the number of RF chains is set larger than the average number of UEs per cell, i.e., 
\begin{equation}
\Lambda = \lambda_u/\lambda_b,
\end{equation}
to ensure that each UE can be served. When the RF chains are insufficient, excess UEs are deferred randomly. More efficient scheduling schemes can be adopted in practice. Let \(x_u\) denote the transmitted symbol of \(\mathrm{UE}_u\), which satisfies \(\mathbb{E}\{|x_u|^2\} = 1\). The uplink received signal at \(\mathrm{BS}_0\) can be written as
\begin{equation}\label{eq:y}
\mathbf{y} =\underbrace{\sum_{\mathrm{U}_u\in\mathcal{U}_0} \mathbf{h}_u x_u}_{\text{intra-cell signals}}+\underbrace{\sum_{\mathrm{U}_v\in\mathcal{U}\setminus\mathcal{U}_0}\mathbf{h}_v x_v}_{\text{ICI (inter-cell interference)}}+\underbrace{\mathbf{n}}_{\text{noise}},
\end{equation}
where \(\mathbf{n} \sim \mathcal{CN}(\mathbf{0}, P_N\mathbf{I})\) is the additive Gaussian noise. 

\vspace{-2mm}
\subsection{Uplink CUMA Scheme}
In the considered system model, only the CSI of the UEs associated with the same cell is assumed to be available at each BS, while the ICI is unknown. As a result, conventional uplink receiver schemes such as ZF, MMSE and MRC are ineffective in suppressing ICI. This motivates us to adopt the CUMA scheme, which was proposed in \cite{CUMA1,CUMA_SG}, where the detailed algorithmic development can be found. Here, we outline the steps essential for subsequent performance analysis.

At the BS, the uplink signals of the associated UEs are decoded separately via a process consisting of port selection and signal superposition. Let \(\mathbf{A}\) denote the binary port-selection matrix for \(\mathrm{BS}_0\), which is expressed as
\begin{equation}
\mathbf{A} =\left[\mathbf{a}_1,\mathbf{a}_2,\dots,\mathbf{a}_K\right],
\end{equation}
where \(\mathbf{a}_m\in\{\mathbf{e}_1,\dots,\mathbf{e}_N\}\) with \(\mathbf{e}_n\) being the standard basis vector, and \(K\) denotes the number of activated FAS ports. For \(\mathrm{UE}_u\), let \(\mathbf{b}_u\in \{0,1\}^{K\times 1}\) denote its switching vector. Then the received signal for \(\mathrm{U}_u\) is given by
\begin{equation}
z_u = \mathbf{b}_u^T\mathbf{A}^T\mathbf{y}.
\end{equation}

We adopt the CUMA strategy and focus on maximizing the power of the desired UE signal. For uplink, we can use the same strategy to determine the activated port set \(\mathcal{K}_u\) for each UE, say \(\mathrm{UE}_u\in\mathcal{U}_0\) separately. Then the overall activated set is obtained from the union of all port sets, i.e.,
\begin{equation}
\mathcal{K} = \bigcup_{\mathrm{U}_u\in\mathcal{U}_0} \mathcal{K}_u.
\end{equation}

The approach to obtain \(\mathcal{K}_u\) is presented as follows. First, we partition all ports according to the sign of their channel coefficients' real parts into two sets: 
\begin{equation}
\left\{\begin{aligned}
\mathcal{K}_u^{+}&=\big\{k:\Re\{\widetilde{h}_u^{(k)}\}\ge0\big\},\\
\mathcal{K}_u^{-}&=\big\{k:\Re\{\widetilde{h}_u^{(k)}\}<0\big\}.
\end{aligned}\right.
\end{equation}
After grouping, signal superposition of the ports within the same group leads to a constructive combination on the complex plane, thereby increasing the magnitude of their sum. Thus, by selecting one group and activating all its ports, each port contributes positively to the desired signal power. Then we can choose the set for providing a larger amplification of signal power, i.e., the selected port set is given by
\begin{equation}
\mathcal{K}_u=\begin{cases}
\mathcal{K}_u^{+}, &\text{if }\sum_{k\in\mathcal{K}_u^{+}}\Re\{\widetilde{h}_u^{(k)}\}
\ge\sum_{k\in\mathcal{K}_u^{-}}\big|\Re\{\widetilde{h}_u^{(k)}\}\big|,\\
\mathcal{K}_u^{-}, &\text{otherwise.}
\end{cases}
\end{equation} 

For conciseness, we define \(\mathbf{w}_u\) to combine the combining vector \(\mathbf{b}_u\) and the port selection matrix \(\mathbf{A}\), i.e., \(\mathbf{w}_u = \mathbf{A}\mathbf{b}_u = [\mathbf{w}_u^{(1)},\mathbf{w}_u^{(2)},\dots,\mathbf{w}_u^{(N)}]^T\). With the activated port set, \(\mathcal{K}_u\), for \(\mathrm{UE}_u\), the element of \(\mathbf{w}_u\) is set as
\begin{equation}
w_u^{(n)}=\begin{cases}
1, & n\in\mathcal{K}_u,\\
0, & \text{otherwise}.
\end{cases}
\end{equation}
The received signal for \(\mathrm{UE}_u\) is therefore given by
\begin{equation}\label{eq:zu}
\begin{aligned}
z_u &=\underbrace{\mathbf{w}_u^T\mathbf{h}_ux_u}_{\text{desired signal}}\\
&+\underbrace{\mathbf{w}_u^T\sum_{\mathrm{U}_w\in\mathcal{U}_0\setminus\{U_u\}} \mathbf{h}_w x_w}_{\text{intra-cell interference}}+
\underbrace{\mathbf{w}_u^T\sum_{\mathrm{U}_v\in\mathcal{U}\setminus\mathcal{U}_0}\mathbf{h}_v x_v}_{\text{ICI}}
+\underbrace{\mathbf{w}_u^T\mathbf{n}}_{\text{noise}}.
\end{aligned}
\end{equation}
The received SINR of \(x_u\) can be expressed as
\begin{equation}\label{eq:gamma_cu}
	\mathrm{SINR}_u
	=
	\frac{|\mathbf{w}_u^T\mathbf{h}_u|^2}{
		\displaystyle
		\underbrace{
			\sum_{\mathrm{U}_w\in\mathcal{U}_0\setminus\{U_u\}}
			|\mathbf{w}_u^T\mathbf{h}_w|^2
		}_{\text{intra-cell interference}}
		+
		\underbrace{
			\sum_{\mathrm{U}_v\in\mathcal{U}\setminus\mathcal{U}_0} |\mathbf{w}_u^T\mathbf{h}_v|^2
		}_{\text{ICI}}
		+
		\underbrace{
			P_N
		}_{\text{noise}}
	}.
\end{equation}

During the port-selection process, the overall interference is treated by the random selection and superposition of ports. Consequently, as the number of selected ports, \(K\), increases, the gain of the interference becomes substantially lower than the desired UE's signal gain, thereby indirectly improving the SINR. This strategy is of low complexity and does not require any CSI from outside the Voronoi cell. In addition, this port-selection scheme does not require any phase control or active device. To evaluate this scheme, as well as discuss the impact of system and environment parameters, the performance analysis will be presented in the next section.

\vspace{-2mm}
\section{Performance Analysis}\label{sec:analysis}
In this section, we characterize the analytical performance of the uplink CUMA scheme. By Slivnyak's theorem, we focus on a typical UE, \(\mathrm{UE}_0\), associated with \(\mathrm{BS}_0\). In the considered multi-cell uplink scenario, interference dominates the system performance, and is thus assumed to be much stronger than noise. As a consequence, we first analyze the SIR, which can be written in a compact ratio form
\begin{equation}
\gamma = \frac{S}{I},
\end{equation}
where \(S\) represents the desired-signal power at the CUMA output, and \(I\) denotes the interference power. For analytical tractability, we assume that the desired signal \(S\) and the interference \(I\) are independent. Although this assumption is not strictly valid due to their inherent correlation, it is commonly adopted in the analysis for FAMA, and has been shown to yield accurate approximations in similar settings. In the rest of this section, we first characterize the distributions of \(S\) and \(I\), then derive the distribution of \(\gamma\), and finally analyze the system performance based on the obtained distribution.

\vspace{-2mm}
\subsection{Distribution Characteristics of \(S\)}
The expressions of \(S\) and \(I\) depend on the selected port set \(\mathcal{K}_0\). However, under the correlated Rayleigh fading assumption and if \(N\) is sufficiently large, the two candidate sets \(\mathcal{K}_0^+\) and \(\mathcal{K}_0^-\) tend to contain an equal number of ports, and are selected with equal probability. Therefore, without loss of generality, we can analyze the overall performance by considering only one of them. In the following, we adopt \(\mathcal{K}_0^+\), and the desired signal power \(S\) can be expressed as
\begin{equation}\label{eq:X_def}
S = \left|\sum_{n=1}^{N} X_n^{+}\right|^2,
\end{equation}
where
\begin{subequations}\label{eq:Xn_def}
\begin{align}
X_n^{+} &= \max\{0,X_n\},\\
\mathbf{X} &= [X_1,X_2,\dots,X_N]^T \sim \mathcal{N}\left(\mathbf{0},\frac{\ell(r_0)}{2}\mathbf{J}\right).
\end{align}
\end{subequations}
Here \(\ell(r)=r^{-\eta}\) denotes the large-scale path loss function for the transmission distance \(r\). The large-scale attenuation changes only the scale of \(X\). As a result, according to \cite{CUMA1}, when \(N\to +\infty\), the conditioned random variable \(\sqrt{S} | r_0\) is approximately Gaussian with mean
\begin{equation}\label{eq:mu1}
\mu_1(r_0) = \sqrt{\ell(r_0)}\bar{\mu}_1 =  N\sqrt{\frac{\ell(r_0)}{2\pi}},
\end{equation}
and variance
\begin{equation}\label{eq:var1}
\begin{aligned}
\sigma_1^2(r_0)& = \bar{\sigma}_1^2\ell(r_0)\\ 
&=\frac{\ell(r_0)N}{2}\left(1-\frac{1}{\pi}\right)+ 2\ell(r_0)\sum_{m=2}^{N}\sum_{k=1}^{m-1}{\rm cov}(X_k^+,X_m^+),
\end{aligned}
\end{equation}
in which
\begin{multline}\label{eq:covX+}
{\mathrm{cov}}(X_k^+,X_m^+) =\frac{\left( 1 - \rho_{k,m}^2 \right)^{\frac{3}{2}} }{2\pi}- \frac{1}{4\pi}\\
+ \frac{\rho_{k,m}}{2 \sqrt{2\pi}}\mathcal{W} \left(- \sqrt{\frac{2\rho_{k,m}^2}{1 - \rho_{k,m}^2}},\frac{1}{2}, \frac{1}{2}\right),
\end{multline}
where \(\rho_{k,m} = [\mathbf{J}]_{k,m}\) is the correlation coefficient, and
\begin{multline}
\mathcal{W}(a, b, c) =- \frac{a \, \Gamma\left( \frac{2c + 3}{2} \right)}{\sqrt{2\pi b} \; b^{\frac{2c + 3}{2}}}\, {}_2F_1 \left(\frac{1}{2}, \frac{2c + 3}{2}; \frac{3}{2} ; -\frac{a^2}{2b}\right)\\
+\frac{\Gamma(c + 1)}{2 b^{c + 1}}.
\end{multline}

Conditioned on \(r_0\), we approximate \(S\) by a Gamma distribution via moment matching:
\begin{equation}\label{eq:Gamma_MM}
S\mid r_0\approx\mathrm{Gamma}\left(m(r_0),\,\theta_S(r_0)\right),
\end{equation}
where
\begin{equation}\label{eq:Gamma_MM_params}
\left\{\begin{aligned}
m(r_0)&=\frac{\mathbb{E}\{S\mid r_0\}^2}{\mathrm{Var}(S\mid r_0)},\\
\theta_S(r_0)&=\frac{\mathrm{Var}(S\mid r_0)}{\mathbb{E}\{S\mid r_0\}}.
\end{aligned}\right.
\end{equation}
The conditioned mean and variance of \(X\) can be expressed in terms of \(\mu_1(r_0)\) and \(\sigma_1(r_0)\) such that
\begin{align}
\mathbb{E}\{S\mid r_0\} &= \mu_1^2(r_0)+\sigma_1^2(r_0) = (\bar{\mu}_1^2+\bar{\sigma}_1^2)\ell(r_0),\\
\mathrm{Var}(S\mid r_0) &= 2\sigma_1^4(r_0)+4\mu_1^2(r_0)\sigma_1^2(r_0).
\end{align}
From the definitions of \(\mu_1\) and \(\sigma_1^2\), given in \eqref{eq:mu1} and \eqref{eq:var1}, it can be known that \(\mu_1 \sim \sqrt{\ell(r_0)}\) and \(\sigma_1^2 \sim \ell(r_0)\). As a result, \(\theta_S(r_0) \sim \ell(r_0)\), and \(m(r_0)\) is independent from \(r_0\). Thus, we can use \(m\) to denote \(m(r_0)\) in the following analysis.

\vspace{-2mm}
\subsection{Distribution Characteristics of \(I\)}
The interference term \(I\) consists of both intra-cell and inter-cell interference. However, since the port selection for \(\mathrm{UE}_0\) depends only on its own CSI and is independent of all other UEs, all interfering UEs affect the receiver only through their effective power. Hence, there is no need to distinguish between intra-cell and inter-cell interferers in the subsequent analysis. Consider \(\mathrm{UE}_v\) with distance \(r_v\) to \(\mathrm{BS}_0\). The effective interference power contributed by \(\mathrm{UE}_v\) can be modeled as
\begin{equation}\label{eq:Vx_def}
V_v \triangleq \left|\sum_{n=1}^{N} t_n\,Y_n(v)\right|^2 = |\mathbf{t}^T\mathbf{Y}(v)|^2,
\end{equation}
where \(\mathbf{Y}(v) = [Y_1(v),\dots,Y_{N}(v)]^T\sim\mathcal{N}\left(\mathbf{0},\frac{\ell(r_v)}{2}\mathbf{J}\right)\), and \(\{t_n = 1\{X_n>0\}\}\) is the port selection indicator coefficient. 

\begin{lemma}\label{le:t}
The port selection indicator \(\mathbf{t}\) can be regarded as a correlated Bernoulli vector with parameters
\begin{align}
\mathbb{E}\{t_n\} &= \frac{1}{2},\\
\mathrm{cov}\{t_k,t_m\} &= \frac{1}{2\pi}\arcsin(\rho_{k,m}).
\end{align}
\end{lemma}

\begin{proof}
See Appendix \ref{app:t}
\end{proof}

Similar to the signal power, the large-scale attenuation \(\sqrt{\ell(r_v)}\) is a common factor across all ports of \(\mathbf{h}_v\). As a result, conditioned on \(r_v\), when \(N\to +\infty\), \(\sqrt{V_v} | r_v\) is approximately zero-mean Gaussian with variance \(\sigma_2^2(r_v) = \ell(r_v)\bar{\sigma}_2^2\), with
\begin{equation}\label{eq:var2}
\begin{aligned}
\bar{\sigma}_2^2&=\sum_{k=1}^{N}\sum_{m=1}^{N}\mathbb{E}\{t_kt_m\}\frac{1}{2}\rho_{k,m}\\
&=\frac{1}{4}N+\frac{1}{4}\sum_{m=2}^{N}\sum_{k=1}^{m-1}\rho_{k,m}+\frac{1}{2\pi}\sum_{m=2}^{N}\sum_{k=1}^{m-1}\rho_{k,m}\arcsin(\rho_{k,m}),
\end{aligned}
\end{equation}
and therefore \(V_v \sim \sigma_2^2(r_v)\mathcal{X}_1\) is a scaled Chi-squared random variable. To analyze the distribution characteristic of \(I\), we can first write the conditioned Laplace transform of \(V_v\):
\begin{equation}\label{eq:LVcv}
\mathcal{L}_{V_v \mid r_v}(s \mid r) = \frac{1}{\sqrt{1+2\sigma_2^2(r)s}}.
\end{equation}
The aggregate interference can be expressed as
\begin{equation}
I =\sum_{\mathrm{U}_v\in\mathcal{U}\setminus\{\mathrm{U}_0\}}V_v.
\end{equation}
To proceed, we adopt an equivalent spatial representation. Let \(x_v \in \mathbb{R}^2\) represent the location of \(\mathrm{U}_v\). Then the aggregate interference can be equivalently expressed as
\begin{equation}
I =\sum_{x_v \in \Phi_u \setminus \{x_0\}} V(x_v),
\end{equation}
where \(\Phi_u\) is the PPP of UE locations, and \(V(x_v)\) denotes the interference contribution from the UE located at \(x_v\).

Conditioned on the spatial realization \(\Phi_u\), the small-scale fading coefficients associated with different UEs are independent. Therefore, the interference contributions \(\{V(x_v)\}\) are independent conditioned on \(\Phi_u\), and thus we have
\begin{align}
\mathcal{L}_I(s \mid \Phi_u) &= \mathbb{E}\left\{\exp(-sI)\mid \Phi_u\right\}\notag\\
&=\mathbb{E}\left\{\prod_{x_v\in\Phi_u\setminus\{x_0\}}\exp\!\left(-sV(x_v)\right)\middle|\Phi_u\right\}\notag\\
&=\prod_{x_v\in\Phi_u\setminus\{x_0\}}\mathbb{E}\left\{\exp\!\left(-sV(x_v)\right)\middle|\Phi_u\right\}\notag\\
&=\prod_{x_v\in\Phi_u\setminus\{x_0\}}\mathbb{E}\left\{\exp\!\left(-sV(x_v)\right)\middle| r_v\right\}\notag\\
&=\prod_{x_v\in\Phi_u\setminus\{x_0\}}\mathcal{L}_{V_v\mid r_v}(s\mid r_v).
\end{align}
The third step follows from the conditional independence of \(\{V(x_v)\}\) given \(\Phi_u\), and the fourth step holds because the distribution of \(V(x_v)\) depends on the UE location only through \(r_v\). Taking the expectation with respect to \(\Phi_u\) and applying the probability generating functional of the PPP \cite{SG2}, the Laplace transform of \(I\) can then be obtained as
\begin{align}
\mathcal{L}_I(s)&=\mathbb{E}_{\Phi_u}\left[\prod_{x \in \Phi_u \setminus \{x_0\}}\mathcal{L}_{V \mid r}(s \mid r)\right]\notag\\
&=\exp\!\left(-\lambda_u\int_{\mathbb{R}^2}\left(1 - \mathcal{L}_{V \mid r}(s \mid r)\right)\mathrm{d}x\right).
\end{align}
Substituting the conditional Laplace transform \eqref{eq:LVcv} and converting the spatial integral into polar coordinates, we obtain
\begin{equation}\label{eq:LY_final}
\mathcal{L}_I(s)=\exp\left(-2\pi \lambda_u\int_{0}^{\infty}\left(1 -\frac{1}{\sqrt{1 + 2 \bar{\sigma}_2^2sr^{-\eta}}}\right)r\,\mathrm{d}r\right).
\end{equation}
To simplify the expression, we introduce the lemma below.

\begin{lemma}\label{le:IA}
When \(\eta>2\), the integral
\begin{equation}\label{eq:I_def}
I(A)=\int_{0}^{+\infty}\left(1-(1+A r^{-\eta})^{-1/2}\right) r\,\mathrm{d}r,
\end{equation}
can be derived as
\begin{equation}\label{eq:I_closed}
I(A) = \frac{A^{2/\eta}}{2}\,\frac{\Gamma\left(1-\frac{2}{\eta}\right)\Gamma\left(\frac12+\frac{2}{\eta}\right)}{\sqrt{\pi}}.
\end{equation}
\end{lemma}

\begin{proof}
See Appendix \ref{app:IA}.
\end{proof}

Subsequently, by substituting \eqref{eq:I_closed} into \eqref{eq:LY_final}, we can obtain the Laplace transform of \(I\) as
\begin{equation}\label{eq:LY_closed}
\mathcal{L}_{I}(s)=\exp\left(-4^{\frac{1}{\eta}}\sqrt{\pi}\Gamma\left(1-\frac{2}{\eta}\right)\Gamma\left(\frac12+\frac{2}{\eta}\right)\lambda_u (\bar{\sigma}_2^2s)^{\frac{2}{\eta}}\right).
\end{equation}

\vspace{-2mm}
\subsection{Distribution Characteristics of \(\gamma\)}
Next, by the independence assumption between \(S\) and \(I\), the conditioned probability \(\mathbb{P}(\gamma>\tau\mid r_0)\) can be derived as
\begin{align}
\mathbb{P}(\gamma>\tau\mid r_0) &= \mathbb{P}(S>\tau I\mid I,r_0)\notag\\
&= \mathbb{E}\{1-F_S(\tau I\mid r_0)\}\notag\\
&\approx \mathbb{E}_I\left\{\frac{\Gamma\left(m,\frac{\tau I}{\theta_S(r_0)}\right)}{\Gamma(m)}\right\},
\end{align}
where \(F_S(\tau I\mid r_0)\) is the cumulative density function (CDF) of \(S\), which has been approximated by Gamma distribution. Then by following Alzer's inequality \cite{Alzer}, a tight lower bound of \(\mathbb{P}(\gamma>\tau\mid r_0)\) is given by
\begin{equation}
\frac{\Gamma\left(m,\frac{\tau Y}{\theta_S(r_0)}\right)}{\Gamma(m)} \gtrapprox 1 - \left(1-\exp\left(-a_m\frac{\tau Y}{m\theta_S(r_0)}\right)\right)^{m},
\end{equation}
where \(a_m = m\Gamma(1+m)^{-\frac{1}{m}}\). To avoid the infinite operations, we take an asymptotic integer \(m_0=\lceil m\rceil\) as the power. Then the conditioned coverage probability can be obtained using the binomial theorem, yielding
\begin{align}\label{eq:cov_cond}
\mathbb{P}(\gamma&>\tau\mid r_0)\notag\\
&\approx\sum_{k=1}^{m_0}(-1)^{k+1}\binom{m_0}{k}\mathbb{E}_I\left\{\exp\left(-ka_m\frac{\tau Y}{m\theta_S(r_0)}\right)\right\}\notag\\
&\approx\sum_{k=1}^{m_0}(-1)^{k+1}\binom{m_0}{k}\mathcal{L}_{I}\left(s_k(r_0,\tau)\right),
\end{align} 
where 
\begin{equation}\label{eq:sk_def}
s_k(r_0,\tau)=\frac{ka_m\tau}{m\theta_S(r_0)}.
\end{equation}

Next, we take the expectation over the distance \(r_0\):
\begin{equation}\label{eq:cov_uncond}
\mathbb{P}(\gamma>\tau)\approx\int_{0}^{\infty}\mathbb{P}(\gamma>\tau\mid r_0=r)f_{r_0}(r)\mathrm{d}r,
\end{equation}
where \(f_{r_0}(r)\) denotes the distribution of \(r_0\). Under the nearest-BS association rule, the serving distance \(r_0\) corresponds to the distance from a typical UE to its nearest point in \(\Phi_b\), whose probability density function (PDF) is given by
\begin{equation}\label{eq:fr0}
f_{r_0}(r) = 2\pi \lambda_b r \exp\!\left(-\pi \lambda_b r^2\right), \ r \geq 0.
\end{equation}

By substituting \eqref{eq:LY_final}, \eqref{eq:cov_cond}, \eqref{eq:sk_def} and \eqref{eq:fr0} into \eqref{eq:cov_uncond}, the coverage probability can be expressed as
\begin{equation}\label{eq:cov_uncond2}
\mathbb{P}(\gamma>\tau)\approx2\pi\lambda_b\sum_{k=1}^{m_0}(-1)^{k+1}\binom{m_0}{k}Z_k,
\end{equation}
where 
\begin{multline}
Z_k = \int_{0}^{+\infty}r\times\\
\exp\left(-\pi\left[\begin{array}{l}
\lambda_br^2+\\
\frac{4^{\frac{1}{\eta}}\lambda_u}{\sqrt{\pi}}\Gamma\left(1-\frac{2}{\eta}\right)\Gamma\left(\frac12+\frac{2}{\eta}\right)(\bar{\sigma}_2^2s_k)^{\frac{2}{\eta}}
\end{array}\right]\right)dr.
\end{multline}

\begin{lemma}\label{le:Zk}
The expression of \(Z_k\) is given by
\begin{equation}\label{eq:Zk}
Z_k = \frac{1}{2\pi G_k(\tau)},
\end{equation}
where
\begin{equation}\label{eq:Gk}
G_k(\tau) = \lambda_b+\frac{4^{\frac{1}{\eta}}\lambda_u}{\sqrt{\pi}}\Gamma\left(1-\frac{2}{\eta}\right)\Gamma\left(\frac12+\frac{2}{\eta}\right)
(\bar{\sigma}_2^2\bar{s}_k(\tau))^{\frac{2}{\eta}},
\end{equation}
and
\begin{equation}\label{eq:bsk}
\bar{s}_k(\tau) = \frac{ka_m\tau}{m\theta_S(1)}.
\end{equation}
\end{lemma}

\begin{proof}
See Appendix \ref{app:Zk}
\end{proof}

Based on these results, a closed-form approximation for the coverage probability can be found as
\begin{equation}\label{eq:cov}
\mathbb{P}(\gamma>\tau)\approx\lambda_b\sum_{k=1}^{m_0}(-1)^{k+1}\binom{m_0}{k}\frac{1}{G_k(\tau)},
\end{equation}
and therefore, the CDF \(F_{\gamma}(\tau)\) is approximated by
\begin{equation}\label{eq:CDF}
F_{\gamma}(\tau) \approx 1 - \lambda_b\sum_{k=1}^{m_0}(-1)^{k+1}\binom{m_0}{k}\frac{1}{G_k(\tau)}.
\end{equation}

\vspace{-2mm}
\subsection{Average rate}
The average rate is a commonly-used performance evaluation metric, which is defined as
\begin{equation}
\bar{R} = \int_{0}^{+\infty}\log(1+\tau)dF_{\gamma}(\tau).
\end{equation}
Given the expression of \(F_{\gamma}(\tau)\), \(\bar{R}\) can be derived as
\begin{equation}
\begin{aligned}
\bar{R} &=\log(1+\tau)F_{\gamma}(\tau)\big|_{\tau=0}^{\tau=\infty}-\int_{0}^{+\infty}F_{\gamma}(\tau)d\log(1+\tau)\\
&= \log(1+\tau)F_{\gamma}(\tau)\big|_{\tau=0}^{\tau=\infty} - \int_{0}^{+\infty}\frac{1}{1+\tau}d\tau\\
&\quad+\lambda_b\sum_{k=1}^{m_0}(-1)^{k+1}\binom{m_0}{k}\int_{0}^{+\infty}\frac{1}{1+\tau}\frac{1}{G_k(\tau)}d\tau\\
&=\lim\limits_{\tau\to\infty}\left[\log(1+\tau)F_{\gamma}(\tau) - \log(1+\tau)F_{\gamma}(\tau)\right]\\
&\quad+\lambda_b\sum_{k=1}^{m_0}(-1)^{k+1}\binom{m_0}{k}\int_{0}^{+\infty}\frac{1}{1+\tau}\frac{1}{G_k(\tau)}d\tau.
\end{aligned}
\end{equation}

\begin{lemma}\label{le:lim}
When \(\eta \geq 2\), we have
\begin{equation}
\lim\limits_{\tau\to\infty}\left[\log(1+\tau)F_{\gamma}(\tau) - \log(1+\tau)\right] = 0.
\end{equation}
\end{lemma}

\begin{proof}
See Appendix \ref{app:lim}.
\end{proof}

Now, rewrite \(G_k(\tau)\) as
\begin{equation}
G_k(\tau) = \lambda_b+\lambda_uH_k\tau^{\frac{2}{\eta}},
\end{equation}
where
\begin{equation}
H_k = \frac{4^{\frac{1}{\eta}}}{\sqrt{\pi}}\Gamma\left(1-\frac{2}{\eta}\right)\Gamma\left(\frac12+\frac{2}{\eta}\right)\left(\bar{\sigma}_2^2\frac{ka_m}{m\theta_S(1)}\right)^{\frac{2}{\eta}}
\end{equation}
is independent from \(\eta\). Then \(\bar{R}\) can be derived from
\begin{equation}\label{eq:aR}
\bar{R} = \sum_{k=1}^{m_0}(-1)^{k+1}\binom{m_0}{k}L\left(\frac{\lambda_u}{\lambda_b}H_k,\frac{2}{\eta}\right),
\end{equation}
where the integration \(L(a,b)\) is denoted by
\begin{equation}\label{eq:Lab}
L(a,b) = \int_{0}^{+\infty}(1+x)^{-1}(1+ax^b)^{-1}dx.
\end{equation}

It is difficult to obtain a closed-form expression for \(L(a,b)\) with an arbitrary \(b\). However, when the path loss scale \(\eta\) is already known, i.e., \(b = \frac{2}{\eta}\) is already known, the integration \(L\left(a,\frac{2}{\eta}\right)\) can be derived using partial fraction decomposition. In general, the value of \(\eta\) remains relatively stable within a given region. Therefore, the derived expression is practically meaningful and can provide useful design insights.

Denoting the average number of UEs per-cell as
\begin{equation}
\mathbb{E}\{U_c\} = \frac{\lambda_u}{\lambda_b} = \Lambda,
\end{equation}
the average sum-rate of one cell can then be derived as
\begin{equation}\label{eq:Rcell}
\begin{aligned}
\bar{R}_{\mathrm{cell}} &= \mathbb{E}\{U_c\}\bar{R} = \Lambda\bar{R}\\
&= \sum_{k=1}^{m_0}(-1)^{k+1}\binom{m_0}{k}\int_{0}^{+\infty}\frac{1}{1+x}\times\frac{1}{\frac{1}{\Lambda}+H_kx^{\frac{2}{\eta}}}dx.
\end{aligned}
\end{equation}

Now we can discuss the impact of \(\Lambda\) on \(\bar{R}_{\mathrm{cell}}\). From \eqref{eq:Rcell}, we know that \(\bar{R}_{\mathrm{cell}}\) increases monotonically with \(\Lambda\), showing that the sum-rate per-cell improves as the number of UEs grows. However, this growth will eventually reach saturation.

\begin{corollary}\label{coro:C1}
When \(\Lambda\to\infty\), the sum-rate per-cell satisfies
\begin{equation}
\bar{R}_{\mathrm{cell}} \to  \bar{R}_{\mathrm{cell},asy} =\frac{\pi}{\sin\left(\frac{2\pi}{\eta}\right)}\sum_{k=1}^{m_0}(-1)^{k+1}\binom{m_0}{k}\frac{1}{H_k}.
\end{equation}
\end{corollary}

\begin{proof}
See Appendix \ref{app:C1}.
\end{proof}

The above asymptotic result, \(\bar{R}_{\mathrm{cell},asy}\), can be regarded as an indicator for evaluating potential capacity, which depends on only path-loss scale \(\eta\) and the parameters of FAS, i.e., \(W\) and \(N\). In addition, the convergence speed of \(\bar{R}_{\mathrm{cell}}\) towards \(\bar{R}_{\mathrm{cell},asy}\) can be quantitatively characterized. 

\begin{corollary}\label{coro:C2}
The convergence speed of \(\bar{R}_{\mathrm{cell}}\) towards \(\bar{R}_{\mathrm{cell},asy}\) can be approximately expressed as
\begin{equation}\label{eq:converge}
\bar{R}_{\mathrm{cell}}=\bar{R}_{\mathrm{cell},asy}-\mathcal{O}\left(\Lambda^{\,1-\frac{\eta}{2}}\right),\ \Lambda\to\infty.
\end{equation}
\end{corollary}

\begin{proof}
See Appendix \ref{app:C2}.
\end{proof}

The results in Corollaries \ref{coro:C1} and \ref{coro:C2} reveal a clear trade-off governed by the path-loss exponent \(\eta\). When \(\eta\) is small, the asymptotic upper bound \(\bar{R}_{\mathrm{cell},asy}\) becomes larger, indicating a higher potential capacity. However, the convergence to this limit is significantly slower, as the gap decays only at the rate \(\mathcal{O}(\Lambda^{1-\frac{\eta}{2}})\). This implies that the system can sustain a high effective load (large \(\Lambda\)) without quickly reaching saturation. In contrast, when \(\eta\) is large, the asymptotic upper bound is smaller, but the convergence is much faster, and the sum-rate saturates quickly. Therefore, the proposed uplink CUMA scheme is particularly advantageous in scenarios with relatively small \(\eta\), where the system operates in a high-load regime and the performance gain can be better exploited.

\vspace{-2mm}
\section{Simulation Results}\label{sec:result}
In this section, we present extensive simulation results to validate our analytical results and compare the proposed uplink CUMA system with several benchmark methods.

\vspace{-2mm}
\subsection{Simulation Setup}
The network model follows the same setup introduced in Section \ref{sec:model}. Specifically, the carrier frequency is set to \(4\) GHz, corresponding to a carrier wavelength of \(\lambda = 7.5\)cm. At each BS, the number of RF chains is chosen as \(\lceil 1.5\Lambda \rceil\), so that the majority of in-cell UEs can be accommodated. Whenever the number of in-cell UEs is larger than the available RF chains, only \(\lceil 1.5\Lambda \rceil\) UEs are randomly scheduled for service, while the remaining in-cell UEs are still assumed to be active and are therefore counted as interference. Also, all fluid antennas are configured as square apertures, i.e., \(W_1 = W_2\), and the numbers of ports along the two dimensions are selected to form an approximately square array, namely \(N_1 \approx N_2\). Unless otherwise specified, the path-loss parameters used in the simulations are set as \(\eta = 2.4\) and \(2.7\), which are taken from the field measurements reported for Austin in \cite{pathloss}.

For Monte-Carlo simluations, two layers of randomness are considered. The first one corresponds to the random realization of BS and UE locations. Since under the PPP assumption the numbers of BSs and UEs over the entire plane are infinite, in simulations, however, we truncate the network to a sufficiently large circle centered at the typical BS and only retain the nodes inside this circle. Nodes outside the circle are neglected because of their large distances and hence weak contributions. In our simulations, the truncation radius is set to a sufficiently large multiple of the average cell size. Specifically, we choose
\begin{equation}
r_{\max} = 20 \times \frac{1}{2\sqrt{\lambda_b}} = 1000\ \text{m},
\end{equation}
in which \( \frac{1}{2\sqrt{\lambda_b}} \) is the average distance from a typical UE to its serving BS. The location realization loop is repeated \(10\) times. This number is not chosen to be very large due to the considerable computational cost, although a larger value would further improve numerical accuracy. The second layer of randomness comes from small-scale fading generation. For each realization of BS and UE locations, the channel coefficients are independently generated over \(10^3\) fading realizations.

Besides the proposed uplink CUMA scheme, three benchmark schemes are considered, namely ZF, full-CSI MMSE, and joint antenna-position and beam optimization (JAPBO). In JAPBO, the BS is assumed to know the CSI of all UEs. In addition to beamformer design, the antenna positions are optimized over the entire set of candidate port locations so as to maximize the resulting SIR. After the antenna positions are determined, the same MMSE beamformer as above is applied. The antenna position update strategy follows from \cite{MA}. Specifically, starting from a random initial configuration, each antenna sequentially examines its four neighboring candidate positions and moves to a new position only if an SIR improvement is obtained. Otherwise, it stays at its current location. This procedure is repeated until convergence. JAPBO incurs a very high computational complexity and is therefore only used as an upper bound for performance evaluation.

The number of antennas of benchmark schemes is set equal to the number of RF chains. Except for JAPBO, all benchmark schemes use fixed-position antennas. For the full-CSI MMSE method, the BS is assumed to have access to the CSI of all UEs, including those outside the cell, which corresponds to an idealized case. The receive beamformer is then designed as
\begin{equation}
\mathbf{w}_u = \left(\mathbf{H}_{\mathrm{in}}\mathbf{H}_{\mathrm{in}}^{H} + \mathbf{H}_{\mathrm{out}}\mathbf{H}_{\mathrm{out}}^{H}\right)^{-1}\widetilde{\mathbf{h}}_u,
\end{equation}
where \(\mathbf{H}_{\mathrm{in}}\) and \(\mathbf{H}_{\mathrm{out}}\) denote the channel matrices of the in-cell and out-of-cell UEs, respectively. 

\vspace{-2mm}
\subsection{Analytical Validation}
Fig.~\ref{fig:CDF} illustrates the CDF of the SIR for different path-loss exponents. We observe that the analytical results in (\ref{eq:CDF}) closely match the empirical curves across the entire range of the SIR threshold \(\tau\). The gap remains small even in the moderate-to-high SIR regime, confirming the effectiveness of the approximations. Although the analytical CDF is derived by using moment matching and Alzer's inequality, it is able to accurately capture the overall distribution of the SIR. As expected, a larger path-loss exponent leads to improved SIR performance due to stronger attenuation of ICI.

\begin{figure}
\centering
\includegraphics[width=.85\linewidth]{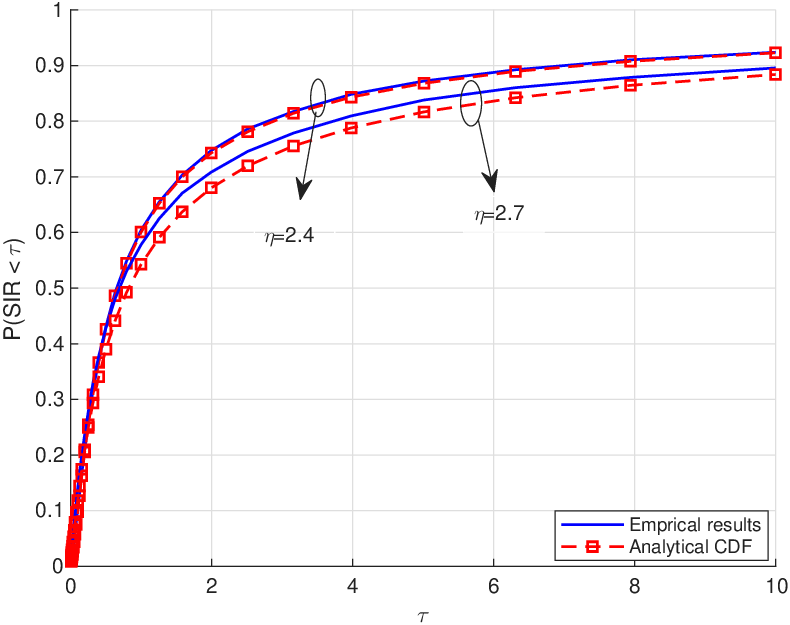}
\caption{Results for the CDF of the SIR.}\label{fig:CDF}
\vspace{-2mm}
\end{figure}

In Fig.~\ref{fig:Sim_L}, we compare the analytical and simulation results in terms of the user average rate as a function of the BS-to-UE density ratio, \(\Lambda\). The analytical results are generated from \eqref{eq:aR}. It is observed that the analytical curves follow the same trend as the simulation results, showing a monotonic decrease of the average rate as \(\Lambda\) increases. This behavior is expected as a higher user density results in stronger intra-cell and ICI as well as increased resource competition. Moreover, the analytical results slightly overestimate the simulation results, which is mainly due to the use of Alzer's inequality and independence assumption in the derivation. In addition, numerical evaluation of the definite integral of (\ref{eq:Lab}) also introduces some errors. Nevertheless, the gap remains relatively small, indicating that the proposed analytical model can provide a reliable approximation of the actual system performance.

\begin{figure}
\centering
\includegraphics[width=.85\linewidth]{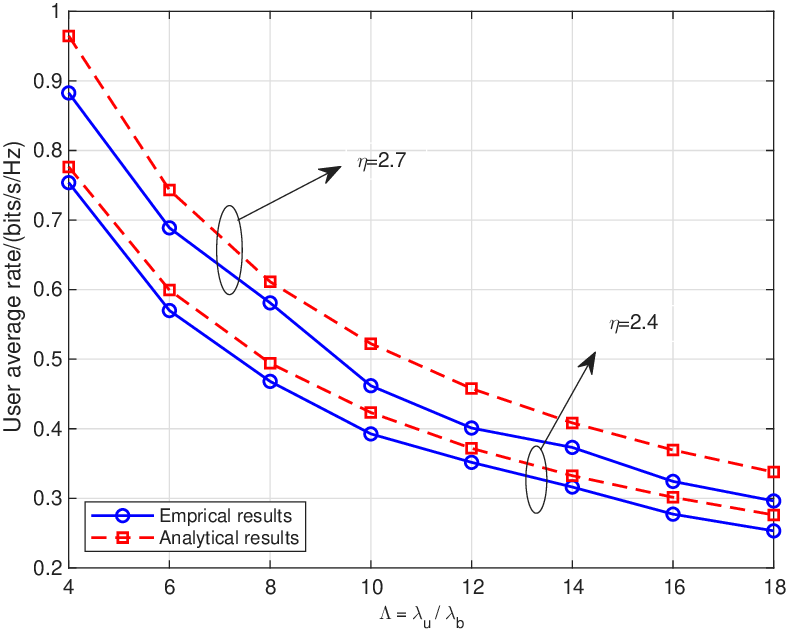}
\caption{Results for the average rate vs.~the BS-UE density ratio, \(\Lambda\).}\label{fig:Sim_L}
\vspace{-2mm}
\end{figure}

Fig.~\ref{fig:asy} validates the asymptotic behavior of the analytical results. To better validate the asymptotic results in Corollaries \ref{coro:C1} and \ref{coro:C2}, relatively large values of the path-loss exponent \(\eta\) are intentionally considered here to provide more diverse behaviors, rather than the practical values used in other figures. As we can see, the analytical results match closely with the limits \(\bar{R}_{\mathrm{cell},asy}\) as \(\Lambda\) increases. Additionally, the convergence speed towards the asymptotic value varies significantly with \(\eta\), consistent with Corollary~\ref{coro:C2}. Specifically, for larger \(\eta = 4\), the convergence speed is about \(\Lambda^{-1}\) according to \eqref{eq:converge}, which matches the observation from Fig.~\ref{fig:asy} that the sum-rate quickly approaches its asymptotic limit. In contrast, for smaller \(\eta = 3.4\) and \(\eta = 2.4\), the convergence speeds are \(\Lambda^{-0.7}\) and \(\Lambda^{-0.2}\), respectively. This confirms that the uplink CUMA scheme is particularly advantageous in cases with relatively small \(\eta\).

\begin{figure}
\centering
\includegraphics[width=.85\linewidth]{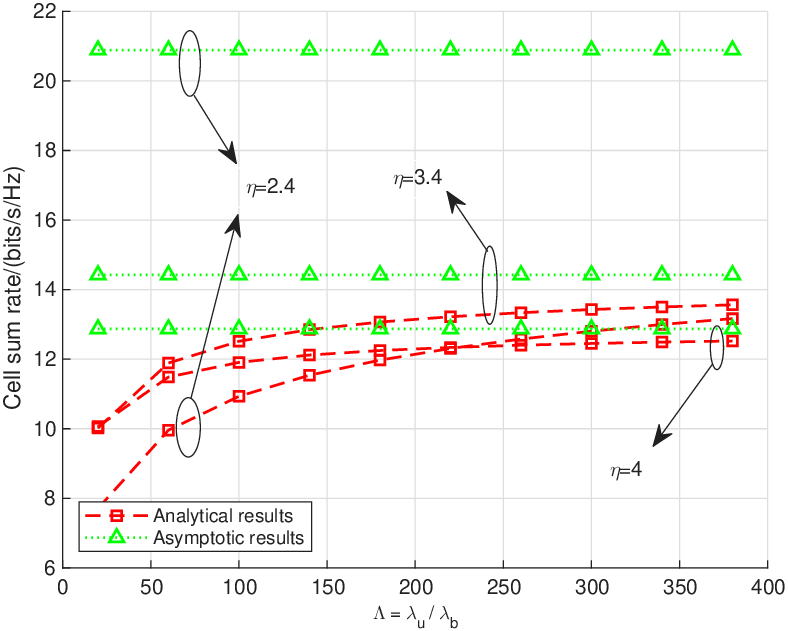}
\caption{Results for the cell sum rate vs.~the BS-UE density ratio, \(\Lambda\).}\label{fig:asy}
\vspace{-2mm}
\end{figure}

\vspace{-2mm}
\subsection{Benchmark Comparison}
Fig.~\ref{fig:vsMMSE_N} shows the cell sum rate as a function of the number of ports \(N\). It can be observed that the performance of ZF and full-CSI MMSE remains almost unchanged as \(N\) increases because both schemes employ a fixed number of antennas determined by the RF chains and hence do not benefit from additional candidate ports. In contrast, CUMA exhibits noticeable fluctuations with respect to \(N\). This behavior originates from the correlation structure in \eqref{eq:Jk}, where the entries depend on Bessel functions, which are inherently oscillatory. As a result, increasing \(N\) does not lead to a monotonic performance improvement but instead introduces fluctuations. Despite this, the additional degrees of freedom introduced by a larger \(N\) can still enhance the overall performance of CUMA. For JAPBO, the performance initially improves as \(N\) increases due to the enlarged search space of candidate antenna positions. However, beyond a certain point, the performance saturates and may even degrade. This is because increasing the number of ports only enhances the spatial resolution of candidate positions without providing additional array gain, while also increasing the likelihood of the local search being trapped in suboptimal configurations. Overall, although JAPBO can achieve slightly higher performance in some regimes, CUMA attains comparable performance with significantly lower complexity and without requiring out-of-cell CSI.

\begin{figure}
\centering
\includegraphics[width=.85\linewidth]{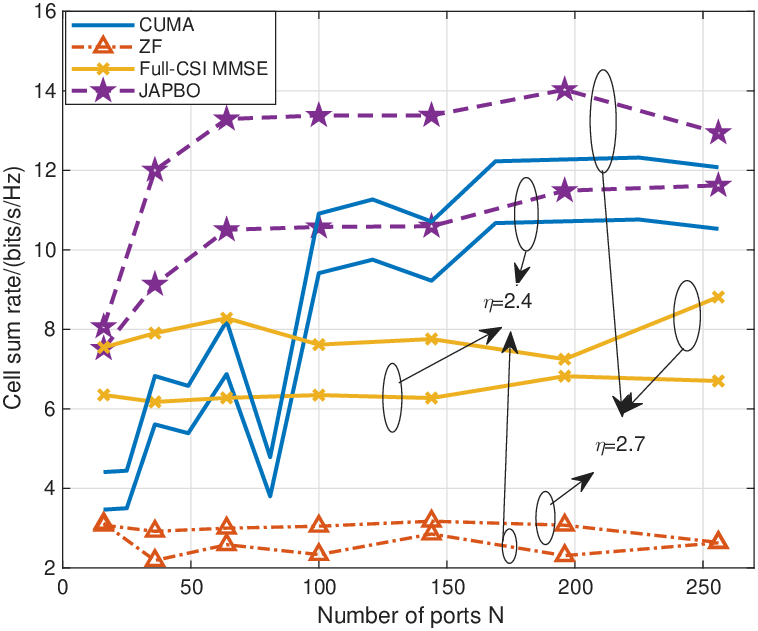}
\caption{Cell sum rate vs.~port count \(N\) with \(\Lambda = 10\), \(W_1 = W_2 = 8\).}\label{fig:vsMMSE_N}
\vspace{-2mm}
\end{figure}

In Fig.~\ref{fig:vsMMSE_L}, we illustrate the cell sum rate versus the BS-to-UE density ratio, \(\Lambda\). As \(\Lambda\) increases, all schemes experience performance degradation due to the increased number of active users and thus the resulting interference. In particular, the advantage of CUMA becomes less pronounced at larger \(\Lambda\). This is because a higher user density reduces the relative impact of ICI, making intra-cell interference the dominant factor. Since CUMA is specifically designed to operate without inter-cell CSI, its relative gain diminishes in this regime. On the other hand, when \(\Lambda\) is small (e.g., \(\Lambda \leq 6\)), CUMA even outperforms JAPBO. This can be attributed to two factors. First, JAPBO is not globally optimal and its performance is inherently limited by the finite number of RF chains. Second, in this regime, ICI dominates and therefore, CUMA can fully exploit its advantages to achieve significant gains. Overall, CUMA still maintains competitive performance compared to the benchmarks while preserving its key advantages in terms of low complexity and reduced CSI requirements.

\begin{figure}
\centering
\includegraphics[width=.85\linewidth]{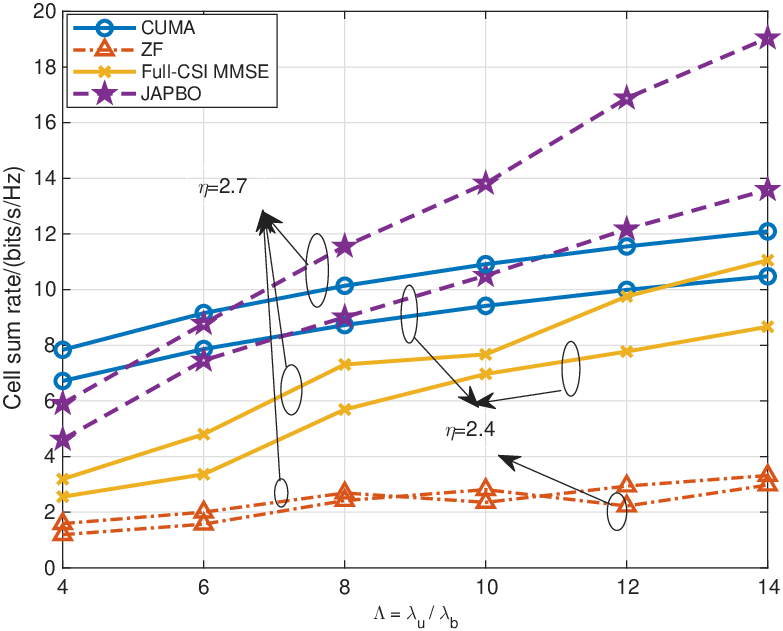}
\caption{Cell sum rate vs.~the BS-UE density ratio, \(\Lambda\), with \(N = 100\), \(W_1 = W_2 = 8\).}\label{fig:vsMMSE_L}
\vspace{-2mm}
\end{figure}

Fig.~\ref{fig:vsMMSE_W} presents the cell sum rate against the FAS width \(W_1\). Similar to the case of varying \(N\), CUMA exhibits pronounced fluctuations as \(W_1\) increases, which again stems from the oscillatory nature of the Bessel function in the correlation model. Compared to the variation with respect to \(N\), the fluctuations with respect to \(W_1\) are even more significant, indicating that the aperture size has a strong impact on the effective channel characteristics. This suggests that simply increasing the FAS size does not necessarily lead to performance improvement. Instead, selecting an appropriate aperture size is crucial for achieving good performance with CUMA. This observation highlights an interesting design trade-off and suggests that the joint optimization of aperture size and port configuration could be a promising direction for future research.

\begin{figure}
\centering
\includegraphics[width=.85\linewidth]{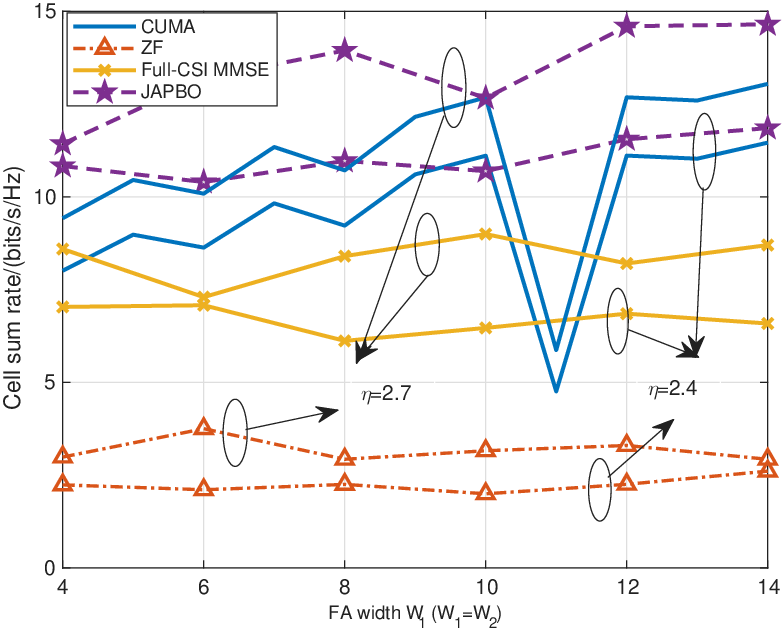}
\caption{Cell sum rate vs.~the FAS width \(W_1\) with \(W_1 = W_2\), \(\Lambda = 10\), \(N = 144\).}\label{fig:vsMMSE_W}
\vspace{-2mm}
\end{figure}

Fig.~\ref{fig:vsMMSE_eta} shows the cell sum rate versus the path-loss exponent \(\eta\). It is observed that the performance of all schemes improves as \(\eta\) increases, due to the stronger attenuation of interference. At the same time, the performance gap between CUMA and the benchmark schemes becomes smaller for larger \(\eta\). This is because, similar to the case of increasing \(\Lambda\), the relative importance of ICI decreases, and intra-cell interference becomes dominant. As such, the advantage of CUMA in handling unknown ICI is less pronounced. Nevertheless, CUMA continues to offer a favorable trade-off between performance, implementation complexity, and CSI acquisition, making it an attractive solution for practical large-scale systems.

\begin{figure}
\centering
\includegraphics[width=.85\linewidth]{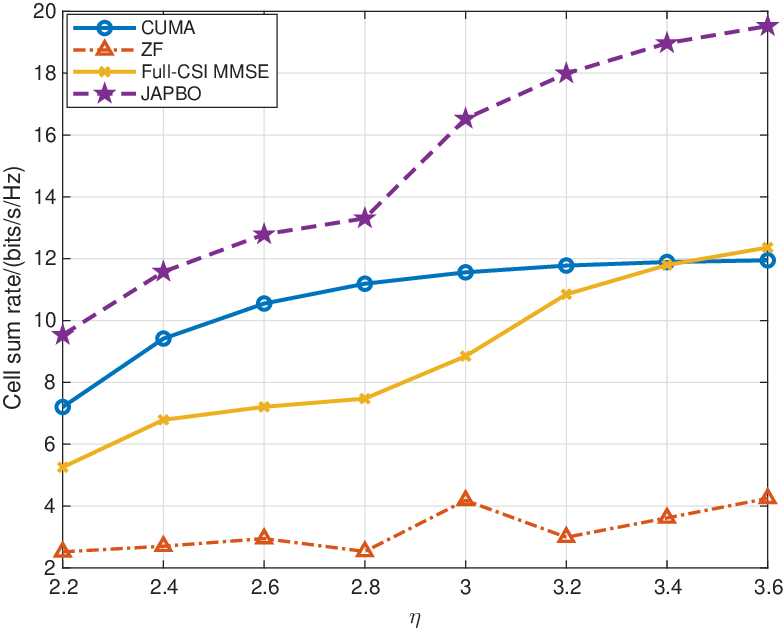}
\caption{Cell sum rate vs.~the path-loss scale \(\eta\) with \(\Lambda = 10\), \(W_1 = W_2 = 8\), \(N = 144\).}\label{fig:vsMMSE_eta}
\vspace{-2mm}
\end{figure}

\vspace{-2mm}
\section{Conclusion}\label{sec:conclude}
In this paper, we investigated CUMA in multi-cell networks. By leveraging a stochastic geometry framework, we developed a tractable analytical approach to characterize the performance of uplink CUMA. In particular, the SIR coverage probability was derived, as well as the average user rate and cell sum-rate. Moreover, the analytical and simulation results demonstrated that uplink CUMA achieves competitive or superior performance compared with some benchmark schemes, and confirmed the potential of uplink CUMA as a low-complexity, low hardware burden and low-CSI uplink reception solution for future large-scale cellular networks. Future work may extend the current analysis to more general channel models, such as scenarios with LoS components or imperfect CSI, and to more advanced system designs, including power control, user scheduling, and hybrid receiver architectures.

\appendices
\section{Proof of Lemma \ref{le:t}}\label{app:t}
From the symmetry of \(X_n\), we know that
\begin{equation}
\mathbb{E}\{t_n\} = \mathbb{P}(X_n > 0) = \frac{1}{2},
\end{equation}
and
\begin{equation}\label{eq:covtktm}
\mathrm{cov}\{t_k,t_m\} = \mathbb{P}(X_k > 0, X_m > 0) - \frac{1}{4}.
\end{equation}
The problem is now converted to obtain \(\mathbb{P}(X_k > 0, X_m > 0)\), i.e., the probability that point \((X_k,X_m)\) is located in the first quadrant in a Cartesian coordinate system. Define two independent variable \(U,V\sim\mathcal{N}(0,1)\), and set
\begin{equation}
\left\{\begin{aligned}
X_k& = U,\\
X_m &= \rho_{k,m}U+\sqrt{1-\rho_{k,m}}V.
\end{aligned}\right.
\end{equation}
It is easy to prove \(X_k,X_m\sim\mathcal{N}(0,1)\) and \(\mathrm{cov}(X_k,X_m) = \rho_{k,m}\). Then the event \(\{X_k > 0, X_m > 0\}\) is equivalent to \(\{U > 0, V > -\alpha_{k,m} U\}\), where
\begin{equation}
\alpha_{k,m} = \frac{\rho_{k,m}}{\sqrt{1-\rho_{k,m}}}.
\end{equation}
Consider the point \((U,V)\) in a Cartesian coordinate system. Denote its coordinate in a Polar coordinate system as \((\Delta,\Theta)\). According to the rotational symmetry of \((U,V)\), \(\Theta\) can be known to be independent with \(\Delta\), and follows the uniform distribution \(\theta\sim \mathcal{U}(-\pi,\pi)\). Then we have
\begin{equation}
\begin{aligned}
\mathbb{P}\{U > 0, V > -\alpha_{k,m} U\} &= \mathbb{P}\left\{-\arctan(\alpha_{k,m})<\theta<\frac{\pi}{2}\right\}\\
&=\frac{\frac{\pi}{2}+\arctan(\alpha_{k,m})}{2\pi}\\
&=\frac{1}{4}+\frac{1}{2\pi}\arcsin(\rho_{k,m}).
\end{aligned}
\end{equation}
Substituting this result into \eqref{eq:covtktm}, the lemma is proved.

\section{Proof of Lemma \ref{le:IA}}\label{app:IA}
Consider a finite integration
\begin{align}
I(A,\Omega_L,\Omega_H)&=\int_{\Omega_L}^{\Omega_H}\left(1-(1+A r^{-\eta})^{-1/2}\right) r\,\mathrm{d}r\notag\\
&=\int_{\Omega_L}^{\Omega_H}rdr-\int_{\Omega_L}^{\Omega_H}\frac{r}{(1+Ar^{-\eta})^{\frac{1}{2}}}\notag\\
&=I_1-I_2.\label{eq:I_LH}
\end{align}
The first term can be easily obtained by
\begin{equation}
I_1 = \frac{\Omega_H^2-\Omega_L^2}{2}.
\end{equation}
To solve the second term, we apply a replacement \(r = t^{-\eta}\) first, which leads to
\begin{equation}
I_2 = \frac{1}{\eta}\int_{\Omega_H^{-\eta}}^{\Omega_L^{-\eta}}\frac{t^{-\frac{2}{\eta}-1}}{(1+At)^{\frac{1}{2}}}.
\end{equation} 
Using the following equation from \cite{ITB}:
\begin{equation}
\int_{0}^{u}\frac{x^{\mu-1}}{(1+\beta x)^{\nu}} = \frac{u^{\mu}}{\mu}{}_2F_1\left(\nu,\mu;1+\mu;-\beta u\right),
\end{equation}
we can express the second term as
\begin{multline}
I_2 = -\frac{\Omega_L^2}{2}{}_2F_1\left(\frac{1}{2},-\frac{2}{\eta};1-\frac{2}{\eta};-A\Omega_L^{-\eta}\right)\\
+\frac{\Omega_H^2}{2}{}_2F_1\left(\frac{1}{2},-\frac{2}{\eta};1-\frac{2}{\eta};-A\Omega_H^{-\eta}\right).
\end{multline}
Then \(I(A,\Omega_L,\Omega_H)\) can be expressed as
\begin{multline}\label{eq:ILH1}
I(A,\Omega_L,\Omega_H) = \frac{\Omega_H^2}{2}\left(1-{}_2F_1\left(\frac{1}{2},-\frac{2}{\eta};1-\frac{2}{\eta};-A\Omega_H^{-\eta}\right)\right)\\
-\frac{\Omega_L^2}{2}\left(1-{}_2F_1\left(\frac{1}{2},-\frac{2}{\eta};1-\frac{2}{\eta};-A\Omega_L^{-\eta}\right)\right).
\end{multline}
When \(\eta>2\) and \(\Omega_H\to\infty\), \(-A\Omega_H^{-\eta}\to 0\), according to the Gauss series of hypergeometric function \cite{DLMF}:
\begin{equation}
{}_2F_1(a,b;c;z) = 1 + \frac{ab}{c}z + O(z^2),~\text{when }z\to 0,
\end{equation}
the first term of \eqref{eq:ILH1} is therefore vanished by
\begin{multline}
\frac{\Omega_H^2}{2}\left(1-{}_2F_1\left(\frac{1}{2},-\frac{2}{\eta};1-\frac{2}{\eta};-A\Omega_H^{-\eta}\right)\right)\\
\stackrel{\Omega_H\to\infty}{\longrightarrow}\frac{\Omega_H^2}{2}O(\Omega_H^{-\eta})\to 0, \eta > 2.
\end{multline}

Next we consider the lower limit when \(\Omega_L\to 0\). Applying the summation of hypergeometric function from \cite{ITB}:
\begin{equation}
\begin{aligned}
&{}_2F_1(a,b;c;z)\\ 
&=\frac{\Gamma(c)\Gamma(b-a)}{\Gamma(b)\Gamma(c-a)}(-z)^{-a}{}_2F_1\left(a,a+1-c;a+1-b;\frac{1}{z}\right)\\
&+\frac{\Gamma(c)\Gamma(a-b)}{\Gamma(a)\Gamma(c-b)}(-z)^{-b}{}_2F_1\left(b,b+1-c;b+1-a;\frac{1}{z}\right).
\end{aligned}
\end{equation}
When \(z\to\infty\), the hypergeometric functions in this formula reach \(1\). Notice that when \(\Omega_L\to 0, -A\Omega_L^{-\eta}\to\infty\). As a result, we have the following asymptotic formula: 
\begin{equation}
\begin{aligned}
&\frac{\Omega_L^2}{2}{}_2F_1\left(\frac{1}{2},-\frac{2}{\eta};1-\frac{2}{\eta};-A\Omega_L^{-\eta}\right)\\ 
&\to\frac{\Omega_L^2}{2}\frac{\Gamma(1-\frac{2}{\eta})\Gamma(-\frac{2}{\eta}-\frac{1}{2})}{\Gamma(-\frac{2}{\eta})\Gamma(\frac{1}{2}-\frac{2}{\eta})}A^{-\frac{1}{2}}\Omega_L^{\frac{\eta}{2}}\\
&+\frac{\Omega_L^2}{2}\frac{\Gamma(1-\frac{2}{\eta})\Gamma(\frac{1}{2}+\frac{2}{\eta})}{\Gamma(\frac{1}{2})\Gamma(1)}A^{\frac{2}{\eta}}\Omega_L^{-2}.
\end{aligned}
\end{equation}
When \(\Omega_L\to 0\), the first term vanishes. Then by substituting the remaining into \eqref{eq:ILH1}, we can obtain the expression of \(I(A) = I(A,0,\infty)\) as \eqref{eq:I_def}.

\section{Proof of Lemma \ref{le:Zk}}\label{app:Zk}
Noticing that
\begin{equation}
s_k(r,\tau) = \frac{ka_m\tau}{m}\frac{1}{\theta_S(r)} \sim \frac{1}{\ell(r)}=r^{\eta},
\end{equation}
we can denote
\begin{equation}
s_k(r,\tau) = \bar{s}_k(\tau)r^{\eta},
\end{equation}
where \(\bar{s}_k(\tau)\) is given by \eqref{eq:bsk} and independent from \(r\). Then \(Z_k\) can be expressed as
\begin{equation}\label{eq:Zk_int}
Z_k = \int_{0}^{+\infty}r\exp\left(-\pi r^2G_k(\tau)\right)dr,
\end{equation}
where \(G_k(\tau)\) is given by \eqref{eq:Gk} and also independent from \(r\). The integration \eqref{eq:Zk_int} can be derived as a closed-form expression via the result from \cite{ITB} as
\begin{equation}
Z_k = \frac{1}{2\pi G_k(\tau)}.
\end{equation}

\section{Proof of Lemma \ref{le:lim}}\label{app:lim}
According to the expression of \(F_{\gamma}(\tau)\), we have
\begin{multline}
\log(1+\tau)F_{\gamma}(\tau) - \log(1+\tau)\\
=-\lambda_b\sum_{k=1}^{m_0}(-1)^{k+1}\binom{m_0}{k}\frac{\log(1+\tau)}{G_k(\tau)}.
\end{multline}
While \(G_k(\tau)\sim \tau^{\frac{2}{\eta}}\), the limit can be derived as
\begin{equation}\label{eq:limitation}
\lim\limits_{\tau\to\infty}\frac{\log(1+\tau)}{G_k(\tau)}= \lim\limits_{\tau\to\infty}\frac{\log(1+\tau)}{\tau^{\frac{2}{\eta}}}
=\lim\limits_{\tau\to\infty}\frac{\tau^{1-\frac{2}{\eta}}}{1+\tau}.
\end{equation}
When \(\eta \geq 2\), \(0 \leq 1-\frac{2}{\eta} < 1\), and thus the result of \eqref{eq:limitation} is known to be zero, which proves the lemma.

\section{Proof of Corollary \ref{coro:C1}}\label{app:C1}
When \(\Lambda\to\infty\), \eqref{eq:Rcell} is derived as
\begin{equation}\label{eq:Rcell0}
\bar{R}_{\mathrm{cell}} \approx \sum_{k=1}^{m_0}(-1)^{k+1}\binom{m_0}{k}\frac{1}{H_k}\int_{0}^{+\infty}(1+x)^{-1}x^{-\frac{2}{\eta}}dx.
\end{equation}
The integral can be expressed by Beta-function as
\begin{equation}
\int_{0}^{+\infty}(1+x)^{-1}x^{-\frac{2}{\eta}}dx = B\left(1-\frac{2}{\eta},\frac{2}{\eta}\right) = \frac{\pi}{\sin(\frac{2\pi}{\eta})}.
\end{equation} 
The second step uses the Euler's reflection formula \cite{ITB}. Then by substituting this integral into \eqref{eq:Rcell0}, \eqref{eq:Rcell} can be obtained.

\section{Proof of Corollary \ref{coro:C2}}\label{app:C2}
The gap between \(\bar{R}_{\mathrm{cell}}\) and its asymptotic limit is given by
\begin{equation}
\bar{R}_{\mathrm{cell},asy}-\bar{R}_{\mathrm{cell}}=\sum_{k=1}^{m_0}(-1)^{k+1}\binom{m_0}{k}\Delta_k(\Lambda),
\end{equation}
where
\begin{align}
\Delta_k(\Lambda)& =\int_0^\infty \frac{1}{1+x}\left(\frac{1}{H_kx^{\frac{2}{\eta}}}-\frac{1}{\frac{1}{\Lambda}+H_kx^{\frac{2}{\eta}}}\right)\mathrm{d}x\notag\\
& = \frac{1}{\Lambda H_k}\int_0^\infty\frac{1}{1+x}\,\frac{1}{x^{\frac{2}{\eta}}\left(\frac{1}{\Lambda H_k}+x^{\frac{2}{\eta}}\right)}\mathrm{d}x.
\end{align}
Denote \(\varepsilon_k=\frac{1}{\Lambda H_k}\), and perform the change of variable \(x=\varepsilon_k^{\eta/2}t=(\Lambda H_k)^{-\eta/2}t\). Then we obtain
\begin{equation}
\Delta_k(\Lambda)=\frac{1}{H_k}\varepsilon_k^{\,1/\frac{2}{\eta}-1}\int_0^\infty\frac{1}{1+\varepsilon_k^{\eta/2}t}\frac{1}{t^{\frac{2}{\eta}}(1+t^{\frac{2}{\eta}})}\mathrm{d}t.
\end{equation}
Since \(\varepsilon_k\) goes to zero as \(\Lambda\to\infty\), the integral term converges to a finite constant when \(\eta>2\). Therefore, the behavior of gap when \(\Lambda\to\infty\) is dominated by \(\varepsilon_k\) as
\begin{equation}
\Delta_k(\Lambda)=\mathcal{O}\left(\frac{1}{H_k}\varepsilon_k^{\frac{\eta}{2}-1}\right)=\mathcal{O}\left(\Lambda^{1-\frac{\eta}{2}}H_k^{-\frac{\eta}{2}}\right).
\end{equation}
Hence, the overall convergence behavior is expressed as \eqref{eq:converge}.

\bibliographystyle{IEEEtran}

\end{document}